\documentclass{article}

\usepackage{arxiv}

\usepackage[utf8]{inputenc} 
\usepackage[T1]{fontenc}    
\usepackage{hyperref}       
\usepackage{url}            
\usepackage{booktabs}       
\usepackage{amsfonts}       
\usepackage{nicefrac}       
\usepackage{microtype}      
\usepackage{lipsum}		
\usepackage{graphicx}
\usepackage{natbib}
\usepackage{doi}

\usepackage{amsmath}
\usepackage{amsthm}
\usepackage{algorithm}
\usepackage{enumitem}

\newcommand{\dd}{\mathrm{d}}
\newcommand{\intrange}[2]{[#1\mathbin{:}#2]}
\newcommand{\R}{\mathbb{R}}
\newcommand{\simiid}{\overset{\text{iid}}{\sim}}
\newcommand{\cd}{\mathbin{\mid}}
\newcommand{\var}{\mathrm{var}}

\newcommand{\dNorm}{\mathcal{N}}

\newtheorem{proposition}{Proposition}[section]
\newtheorem{remark}{Remark}[section]
\newtheorem{definition}{Definition}[section]

\newtheorem{example}{Example}[section]

\title{Monte Carlo twisting for particle filters}

\author{ Joshua J Bon \\
	School of Mathematical Sciences\\
	Queensland University of Technology\\
	Brisbane, Australia \\
	\texttt{joshuajbon@gmail.com} \\
	\And
	Christopher Drovandi \\
	School of Mathematical Sciences\\
	Queensland University of Technology\\
	Brisbane, Australia \\
	\texttt{c.drovandi@qut.edu.au} \\
	\And
	Anthony Lee \\
	School of Mathematics\\
	University of Bristol\\
	Bristol, U.K. \\
	\texttt{anthony.lee@bristol.ac.uk} \\
}

\renewcommand{\shorttitle}{Monte Carlo twisting for particle filters}

\hypersetup{
pdftitle={A template for the arxiv style},
pdfsubject={q-bio.NC, q-bio.QM},
pdfauthor={David S.~Hippocampus, Elias D.~Striatum},
pdfkeywords={Random-weight sequential Monte Carlo, Twisted Feynman--Kac model, Rejection sampler.},
}

\begin{document}
\maketitle

\begin{abstract}
We consider the problem of designing efficient particle filters for twisted Feynman--Kac models. Particle filters using twisted models can deliver low error approximations of statistical quantities and such twisting functions can be learnt iteratively. Practical implementations of these algorithms are complicated by the need to (i) sample from the twisted transition dynamics, and (ii) calculate the twisted potential functions. We expand the class of applicable models using rejection sampling for (i) and unbiased approximations for (ii) using a random weight particle filter. We characterise the average acceptance rates within the particle filter in order to control the computational cost, and analyse the asymptotic variance. Empirical results show the mean squared error of the normalising constant estimate in our method is smaller than a memory-equivalent particle filter but not a computation-equivalent filter. Both comparisons are improved when more efficient sampling is possible which we demonstrate on a stochastic volatility model.
\end{abstract}


\section{Introduction}
Particle filters are a class of algorithms that iteratively generate particle approximations to match a set of probability distributions. They are commonly used to estimate filtering or smoothing distributions in general state-space hidden Markov models \citep{doucet2011tutorial}. They also provide unbiased approximations of the normalising constant of a model, which can be used in pseudo-marginal Markov chain Monte Carlo methods \citep{andrieu2009pseudo,andrieu2010particle}, or for model selection and model averaging.

Efficient implementations of particle filters can be defined by altering the dynamics of the model underlying the filter. Look-ahead strategies \citep[see][for review]{lin2013lookahead} incorporate information from future observations in order to direct the particles to high probability regions of the state-space. The optimal look-ahead strategy utilising all future observations changes, or twists, the dynamics of the model to that of the backward information filter in state-space models \citep{bresler1986two,briers2010smoothing}. Such a particle filter is perfect in the sense that it provides an exact estimate of the normalising constant. Recent methods have proposed estimating these optimal twisting functions based on iterative approximations \citep{guarniero2017iterated,heng2020controlled} but are restricted to analytically tractable classes of models and twisting functions.

We develop methods to use twisting functions for particle filters in contexts where it is not currently possible to do so. In particular, models with Markov kernels that have intractable transition densities or more generally where the twisted dynamics cannot be evaluated exactly.

We now define some notation. For some measure $\mu$ and non-negative kernel $K$ on a measurable space $(\mathsf{X}, \mathcal{X})$, we define~$\mu(f) = \int_\mathsf{X}f(x)\mu(\dd x)$ and $K(f)(z) = \int_\mathsf{X}f(x)K(z,\dd x)$ for $z \in \mathsf{X}$, when $f:\mathsf{X} \rightarrow \mathbb{R}$ is measurable. We also write~$\mu K(S) = \int_{\mathsf{X}}\mu(\dd x) K(x, S)$ for $S \in \mathcal{X}$. If~${X_i \simiid \mu}$ then each random variable $X_i$ is independently and identically distributed according to $\mu$.
We denote $\mathcal{C}(a_{1},a_{2}\ldots,a_{m})$ as the categorical distribution with support $\{1,2,\ldots,m\}$ with probabilities proportional to $(a_{1},a_{2},\ldots,a_{m})$ for $a_i \geq 0$.
Let $y_{i:j}$ be the vector $(y_i, y_{i+1},\ldots, y_j)$ when $i<j$ and $y_{k:k} = y_{k}$. We denote the set of integers from $n_1$ up to $n_2$ as $\intrange{n_1}{n_2} = \{n_1,\ldots,n_2\}$ and define~$\prod_{i=1}^{0} y_{i} = 1$.

\section{Background}
\subsection{Discrete-time Feynman--Kac models}
Particle filters traverse a sequence of probability distributions, $\eta_0, \ldots, \eta_n$, that can be described recursively as follows. Consider the measurable space $(\mathsf{X}, \mathcal{X})$ with probability distribution $M_0$, Markov kernels $M_1,\ldots, M_n$, and functions $G_p : \mathsf{X} \rightarrow (0, \infty)$ for $p \in \intrange{0}{n}$ and where~$n$, the terminal time, is a positive integer. The predictive marginal measures, $\gamma_p$, and distributions, $\eta_p$, are defined as
\begin{equation}\label{eq:mmrecur}
    \gamma_{p}(S) = \int_{\mathsf{X}} \gamma_{p-1}(\dd x) G_{p-1}(x)M_{p}(x, S) \quad\text{and}\quad
    \eta_{p}(S) = \frac{\gamma_{p}(S)}{\gamma_{p}(\mathsf{X})}, \quad S \in \mathcal{X}, 
\end{equation} 
whilst the updated marginal measures, $\hat\gamma_p$, and distributions, $\hat\eta_p$, are
\begin{equation*}
    \hat\gamma_{p}(S) = \int_S\gamma_{p}(\dd x) G_{p}(x) \quad\text{and}\quad
    \hat\eta_{p}(S) = \frac{\hat\gamma_{p}(S)}{\hat\gamma_{p}(\mathsf{X})}, \quad S \in \mathcal{X}, 
\end{equation*} 
for $p \in \intrange{1}{n}$ with $\gamma_0 = \eta_0 = M_0$.  We will refer to a Feynman--Kac model as the tuple~$(M_{0:n},G_{0:n})$ which we relate to the terminal measure of the particle filter by noting that
\begin{equation*}
    \gamma_{n}(\varphi) =  E\left\{ \varphi(X_n)\prod_{p=0}^{n-1}G_p(X_p) \right\}, \quad X_0\sim M_0,\; X_p \sim M_p (X_{p-1},\cdot), \quad p \in \intrange{1}{n}.
\end{equation*} 
 
\begin{example}[Hidden Markov model]
A hidden Markov model is a bivariate Markov chain~$(X_p, Y_p)$ for $p \in \intrange{0}{n}$, where the latent variables $X_{0:n}$ are themselves a Markov chain governed by dynamics $X_0 \sim M_0$ and $X_p \sim M_p (X_{p-1},\cdot)$ for $p \in \intrange{1}{n}$. Moreover, the observation variable~$Y_p$ is independent (of all other variables) conditional on ${X_p}$ and admits the density $g_p(x, \cdot)$ given~$\{X_p = x\}$ for $p \in \intrange{0}{n}$. Given observed values $y_{0:(n-1)}$, and defining $G_p(\cdot) = g_p(\cdot, y_p)$ for $p \in \intrange{0}{n}$, the conditional distribution of the latent state $X_n$ is $\eta_n$ and marginal likelihood~$\gamma_n(1)$. An additional observation value,~$y_n$, can be incorporated by considering $\hat\eta_n$ and~$\hat\gamma_n(1)$, respectively.
\end{example}

\subsection{Particle filters}
Particle filters approximate Feynman--Kac models by iteratively generating collections of points, denoted by $\zeta_{p}^{i}$ for $i \in \intrange{1}{N}$, to approximate the sequence of probability measures~$\eta_{p}$ for $p \in \intrange{0}{n}$. The best known of these algorithms is the bootstrap particle filter \citep{gordon1993novel}, described in Algorithm~\ref{alg:bpf}.
\begin{algorithm}
\caption{The Bootstrap Particle Filter}
\label{alg:bpf}
\begin{enumerate}
    \item Sample initial $\zeta^{i}_{0} \simiid \eta_{0}$ for $i \in \intrange{1}{N}$
    \item For each time $p = 1,2,\ldots, n$
    \begin{enumerate}[label=(\roman*)]
        \item Sample ancestors $A^{i}_{p-1} \sim \mathcal{C}\left(G_{p-1}(\zeta^{1}_{p-1}), \ldots, G_{p-1}(\zeta^{N}_{p-1}) \right)$ for $i \in \intrange{1}{N}$
        \item Sample prediction $\zeta^{i}_{p} \sim M_{p}(\zeta^{A^{i}_{p-1}}_{p-1},\cdot)$ for $i \in \intrange{1}{N}$
    \end{enumerate}
\end{enumerate}
\textbf{Output:} Particles $\zeta^{1:N}_p$ for $p = 0,1,\ldots,n$
\end{algorithm}
After running the bootstrap particle filter the approximate predictive measures are
\begin{equation*}
    \eta_{p}^{N} = \frac{1}{N}\sum_{i=1}^{N}\delta_{\zeta^{i}_p}, \quad \gamma_{p}^{N} = Z_{p}^{N}\eta_{p}^{N}, \quad p \in \intrange{0}{n},
\end{equation*}
where $Z_p^{N} \equiv \gamma_p^{N}(1) = \prod_{t=0}^{p-1} \eta_t^{N}(G_t)$. The updated measures are
\begin{align*}
    \hat\eta_{p}^{N} = \sum_{i=1}^{N}W_{p}^{i}\delta_{\zeta^{i}_p}, \quad \hat\gamma_{p}^{N} = \hat Z_{p}^{N}\hat\eta_{p}^{N}, \quad p \in \intrange{0}{n},
\end{align*}
where $W_{p}^{i} =  \{\sum_{j=1}^N G_p(\zeta^{j}_p)\}^{-1}G_p(\zeta^{i}_p)$ and $\hat{Z}_p^{N} \equiv \hat\gamma_p^{N}(1) = \prod_{t=0}^{p} \eta^{N}_t(G_t)$.

The resampling step, 2(i) of Algorithm~\ref{alg:bpf}, need not occur each iteration of the particle filter. It can also occur dynamically, typically by monitoring the effective sample size of the particle population \citep{kong1994sequential,liu1995blind,del2012adaptive}. In this case resampling occurs when the effective sample size drops below $\kappa N$, where $\kappa \in [0,1]$ is user specified. 

\subsection{Twisted Feynman--Kac models}
Twisted Feynman--Kac models \citep{guarniero2017iterated} change the Markov kernels and potential functions of the underlying model whilst preserving the terminal marginal measure and terminal path measure. In particular, for a set of twisting functions $\psi_p : \mathsf{X} \rightarrow (0, \infty)$ where $M_p(\psi_p)$ is bounded for $p \in \intrange{0}{n}$, the Markov kernels undergo the change of measure
\begin{equation}\label{eq:twistmk}
    M^{\psi}_{0}(\dd x_0) = \frac{M_{0}(\dd x_0)\psi_{0}(x_0)}{M_0(\psi_{0})}, \qquad 
    M^{\psi}_{p}(x_{p-1}, \dd x_{p}) = \frac{M_{p}(x_{p-1}, \dd x_{p})\psi_p(x_p)}{M_p(\psi_{p})(x_{p-1})}, 
\end{equation}
for $p \in \intrange{1}{n}$. Whilst the twisted potential functions are
\begin{equation}\label{eq:twistpot}
\begin{split}
    G^{\psi}_{0}(x_0) &= \frac{G_{0}(x_0)}{\psi_{0}(x_{0})}M_1(\psi_{1})(x_{0}) M_{0}(\psi_0),
    \quad
    G^{\psi}_{n}(x_n) = \frac{G_{n}(x_n)}{\psi_n(x_n)}, \\
    G^{\psi}_{p}(x_p) &= \frac{G_{p}(x_p)}{{\psi_{p}(x_{p})}}M_{p+1}(\psi_{p+1})(x_{p}),
\end{split}
\end{equation}
for $p \in \intrange{1}{n-1}$. We focus on twisting where the terminal updated marginal measure and distribution, $\hat\gamma_{n}$ and $\hat\eta_{n}$, are preserved, though it is trivial to preserve their predictive counterparts by setting~$\psi_n(x) = 1$. We denote the marginal measures and distributions of this twisted model with superscript $\psi$, e.g. $\gamma_p^\psi$. For a given Feynman--Kac model, a recursive definition for an optimal sequence of twisting functions is described in Proposition~\ref{prop:opttwist}.
\begin{proposition}\label{prop:opttwist}
Consider a twisted Feynman--Kac model $(M_{0:n}^{\psi^\star},G_{0:n}^{\psi^\star})$ with twisting functions $\psi_{n}^{\star} = G_n$ and $\psi_{p}^{\star} = G_{p} \cdot M_{p+1}( \psi^{\star}_{p+1})$ for $p \in \intrange{0}{n-1}$. The particle approximation of the normalising constant $(\hat{\gamma}_n^{\psi^{\star}})^{N}(1) = \hat{\gamma}_n(1)$ almost surely for any $N \geq 1$. 
\end{proposition}
A proof of Proposition~\ref{prop:opttwist} can be found in \citet{guarniero2017iterated} or \citet{heng2020controlled}. It can also be shown that a twisted Feynman--Kac model is invariant to scaling of the twisting functions. 

Exact usage of the optimal twisting functions is typically intractable since the calculation of the optimal twisting functions are just as difficult as expectations with respect to the original Feynman--Kac model. Moreover, one must be able to generate samples from the twisted mutation functions in \eqref{eq:twistmk}, which may be without a known form  and involve an intractable normalising constant, which is also required for the twisted potential functions.

Despite these difficulties, approximating the optimal twisting functions has proved beneficial for designing efficient particle filters in some important cases. \citet{guarniero2017iterated} and \citet{heng2020controlled} use the recursive definition of the optimal twisting functions (Proposition~\ref{prop:opttwist}) to motivate a parametric approximation of the optimal functions via regression using a previous particle approximation. In this context, the class of Feynman--Kac models and twisting functions considered thus far have generally involved Gaussian Markov kernels and twisting functions consisting of (mixtures of) exponential-quadratic functions. Restricting attention to this class has ensured exact sampling and analytical calculation of the twisted potentials is possible. Our approach, described in the sequel, presents a novel method for utilising twisted models in particle filters when one or both of these assumptions fail.

\section{Particle filters for twisted models}

\subsection{Twisting the Markov kernels}

In order to achieve exact twisting of the Markov kernels, without assuming a particular distributional form, we use rejection sampling. We assume that $\psi_p : \mathsf{X} \to [\epsilon,1]$ with $0< \epsilon <1$ for each $p \in \intrange{0}{n}$. Due to scale-invariance of twisted models this only imposes that the twisting functions are bounded. To sample $X_p \sim M_p^{\psi}(x_{p-1}, \cdot)$, one can sample a candidate $X^\prime \sim M_p(x_{p-1}, \cdot)$ and accept $X_p = X^\prime$ with probability $\psi_p(X^\prime)$, generating independent proposals until the first is accepted. The lower bound, $\epsilon$, ensures a finite acceptance time and bounds the average computation cost by $\epsilon^{-1}$ from above.


Consider a $\omega$-twisted Feynman-Kac model, where $\omega_p : \mathsf{X} \to [\epsilon,1]$ for each $p \in \intrange{0}{n}$. Rejection sampling from $M_p^{\omega}(x_{p-1}, \cdot)$ can be computationally costly if the expected acceptance probability,~$M_p(\omega_p)(x_{p-1})$, is low. As such, we wish to control the average acceptance rate of the rejection sampler at each time $p$ within the particle filter. The idealised average acceptance rates can be expressed as~$\alpha_{0}^{\omega} \equiv M_0(\omega_0)$ and~$\alpha_{p}^{\omega} \equiv \hat{\eta}_{p-1}^{\omega}M_p(\omega_p)$ for $p \in \intrange{1}{n}$. Estimates of these expressions require particle approximations of a $\omega$-twisted model, involving the exact algorithm we are trying to control the acceptance rates for. As such, we seek an alternative representation of the average acceptance rates, which can be estimated from a previous (twisted-model) particle filter. Proposition~\ref{prop:acceptrate} provides a general result that is able to express the average acceptance rates in this form (proof in the supplementary materials).

\begin{proposition}\label{prop:acceptrate}
Consider a rejection sampler for $M^{\omega}_p$ using $M_p^\varrho$ as the proposal distribution with $\sup_{x \in \mathsf{X}}\omega_p(x)\varrho_p(x)^{-1} = 1$, $p \in \intrange{0}{n}$. The average acceptance rate of this sampler within an $\omega$-twisted model $(M_{0:n}^{\omega},G_{0:n}^{\omega})$ can be written with respect to $(M_{0:n}^{\psi},G_{0:n}^{\psi})$ as
\begin{align*}
    \alpha_{0}^{\omega,\varrho} &\equiv M_0^\varrho(\omega_{0} / \varrho_0) = \frac{M_0(\omega_0)}{M_0(\varrho_0)}, \\
     \alpha_{p}^{\omega,\varrho} &\equiv \hat{\eta}_{p-1}^{\omega}M_p^\varrho(\omega_p / \varrho_p) =  \frac{\hat\eta_{p-1}^{\psi}[ M_{p}(\omega_{p})^2 / \{ M_{p}(\psi_{p}) \cdot M_p(\varrho_p)\}]}{\hat\eta_{p-1}^{\psi}\{ M_{p}(\omega_{p}) / M_{p}(\psi_{p})\}}, \quad p \in \intrange{1}{n}.
\end{align*}
\end{proposition}
\begin{remark}
If the proposal distribution is the base Markov kernel, $M_p$, then $\varrho_p = 1$ leads to
\begin{equation*}
\alpha_{0}^{\omega} =  \alpha_{0}^{\omega,1} = M_0(\omega_0), \quad \alpha_{p}^{\omega} = \alpha_{p}^{\omega,1} = \frac{\hat\eta_{p-1}^{\psi}\{ M_{p}(\omega_{p})^2 / M_{p}(\psi_{p}) \}}{\hat\eta_{p-1}^{\psi}\{ M_{p}(\omega_{p}) / M_{p}(\psi_{p})\}}, \quad p \in \intrange{1}{n}.
\end{equation*}
\end{remark}
\begin{remark}
For some Markov kernels, there may exist an exact sampler for $M^\varrho_p$ that we wish to use as the proposal distribution when rejection sampling for $M^\omega_p$. In this case, we accept a proposal $X^\prime \sim M^\varrho_p(x_{p-1}, \cdot)$ when $\omega_p(X^\prime)\varrho_p(X^\prime)^{-1} > U$ and therefore the conditional acceptance rate is~${M_p^\varrho(\omega_p /\varrho_p)}$. Here, the result of Proposition~\ref{prop:acceptrate} directly expresses the average acceptance rates in terms of a $\psi$-twisted model. We refer to this scenario as partial analytical twisting.
\end{remark}
Characterising the average acceptance rates in terms of an alternative model (i.e. the $\psi$-twisted model) allows us to estimate and ultimately control the acceptance rate in the iterative learning procedure detailed in the next section.

\subsection{Learning twisting functions}
We propose to learn a set of twisting functions $\psi_{0:n}^\prime$ from a previous particle approximation according to the backwards recursion given in Proposition~\ref{prop:opttwist}. Initially, the base Feynman--Kac model is used with~$\psi_p(x) = 1$ for $p \in \intrange{0}{n}$ and $x \in \mathsf{X}$. In subsequent iteration improve upon the approximation of the optimal twisting functions using the most recent set of particles.
\begin{algorithm}
\caption{Learning the twisting functions}
\label{alg:learntwist}
Given particles $\zeta_p^{1:N} \sim \eta_p^{\psi}$ for $p \in \intrange{0}{n}$, $\tilde{N}$ Monte Carlo samples, target acceptance rate $\alpha_\text{min}$, and class of approximation functions $\mathsf{H}_p$ for $p \in \intrange{0}{n-1}$
\begin{enumerate}
    \item $\psi^\prime_n(x_n) = G_n(x_n)^{\beta_n}$ with $\beta_n$ chosen to target $\alpha_\text{min}$ 
    \item For each time $p = n-1,\ldots, 1, 0$
    \begin{enumerate}[label=(\roman*)]
        \item Set $\lambda_p^{\tilde{N}}(i) = \frac{G_p(\zeta_{p}^i)}{\psi_p(\zeta_{p}^i)}M_{p+1}^{\tilde{N}}( \psi_{p+1}^{\prime})(\zeta_{p}^{i})$ for $i \in \intrange{1}{N}$
        \item Find $\lambda_p = \arg\min_{h \in \mathsf{H}_p} \sum_{i=1}^N\left(\log\lambda_p^{\tilde{N}}(i)  - \log h(\zeta_{p}^i) \right)^2$
        \item Set $\psi^\prime_p = ( \psi_p \cdot \lambda_p)^{\beta_p}$ with $\beta_p$ chosen to target $\alpha_\text{min}$ 
    \end{enumerate}
\end{enumerate}
\textbf{Output:} $\psi^\prime_{0:n}$ 
\end{algorithm}

Algorithm~\ref{alg:learntwist} is similar to the learning procedures of the iterated auxiliary particle filter \citep{guarniero2017iterated} and controlled sequential Monte Carlo \citep{heng2020controlled} but there are a few notable differences. Our approach generally uses the terminal potential directly as the twisting function at time $n$ and we use a Monte Carlo estimate of $M_{p+1}(\psi^{\prime}_{p+1})$, denoted by $M^{\tilde N}_{p+1}(\psi^{\prime}_{p+1})$, rather than requiring analytical tractability of the integral. The unbiased estimate is
\begin{equation}\label{eq:mintest}
    M^{\tilde N}_{p}(\psi_{p})(x) = \tilde{N}^{-1}\sum_{j=1}^{\tilde N} \psi_p(U^{i}), \quad U^{i} \simiid M_{p}(x, \cdot)
\end{equation}
for some twisting function $\psi_{p}$ and $x \in \mathsf{X}$.
The main difference between our approach and that of \citet{guarniero2017iterated} or \citet{heng2020controlled} is the adjustment of the new twisting function by an inverse temperature $\beta_p$ in Algorithm~\ref{alg:learntwist}. Whilst the lower bound,~$\epsilon$, specifically controls the worst case computational cost, solely relying on this mechanism can result in most or all proposals being accepted via this bound. In this case we are essentially sampling from the untwisted $M_p$.  Using the inverse temperature allows us to target a particular computational cost (via the acceptance rate) whilst retaining some informativeness in the twisting functions. 

Proposition~\ref{prop:acceptrate} allows us to calculate the acceptance rate for a given $\beta_p$, using the particles available in Algorithm~\ref{alg:learntwist}. We detail our strategy for approximating the $\beta_p$ required for a given~$\alpha_\text{min}$ based on this relation in the supplementary materials (Algorithm~\ref{alg:learntemper}).

Moreover, if there exists a family of functions $\mathsf{R}$ for which $M^\varrho_p$ can be sampled from easily when $\varrho_p \in \mathsf{R}$, then we may wish to use partial analytical twisting. In this case, we can choose $\varrho_p$ such that the acceptance rate is maximised by again estimating the rate with Proposition~\ref{prop:acceptrate}.  This method is also presented in the supplementary materials (Algorithm~\ref{alg:maxacc}). Tempering can also be applied if necessary, but if not we simply set $\beta_p = 1$.

\subsection{Effect of approximating the twisted potential functions}

Evaluating the twisted potentials in \eqref{eq:twistpot} for the particle filter requires calculation of integrals~$M_p(\psi_p)$ for $p \in \intrange{0}{n}$. Analytical calculation is not possible, so we use an unbiased approximation of these integrals in our particle filter. In particular we use \eqref{eq:mintest} for $p \in \intrange{1}{n}$ and the analogous estimate for $p = 0$.

We analyse the effect of the twisted potential approximation within the particle filter using sequential Monte Carlo (SMC) with random weights. Random-weight SMC has been considered in a variety of contexts including SMC for estimating parameters of stochastic differential equations \citep{beskos2006exact,fearnhead2010random}. We show that using the approximation \eqref{eq:mintest} leads to unbiasedness of the normalising constant estimate (see supplementary materials). Moreover, we state the asymptotic variance (see supplement for definition) of the random-weight SMC algorithm as follows.

\begin{proposition}\label{prop:asyvar} 
Consider the models $(M_{0:n}, \bar{G}_{0:n})$ and $(M_{0:n}, G_{0:n})$ with marginal measures~$\bar\gamma_{0:n}$ and $\gamma_{0:n}$ respectively. If $\bar{G}_{p}$ is an unbiased estimate of $G_p$ then the asymptotic variance of $\bar\gamma_n^{N}(\varphi)/\gamma_n(1)$ is  $\bar\sigma_n^{2}(\varphi) = \sum_{p=0}^n \bar{v}_{p,n}(\varphi)$ for bounded $\varphi$ where
\begin{equation*}
    \bar{v}_{p,n}(\varphi) = \frac{\gamma_p(1)\gamma_p(s_p \cdot Q_{p,n}(\varphi)^2)}{\gamma_n(1)^2} - \eta_{n}(\varphi)^2,\quad p \in \intrange{0}{n}
\end{equation*}
with $s_p(x_p) = \var\{\bar{G}_p(x_p)/G_p(x_p)\} + 1$ for $p \in \intrange{0}{n-1}$ and $s_n(x_n) = 1$.
\end{proposition}
\begin{remark}
A model with the $\bar G_p = G_p$ reduces to $s_p = 1$ for $p \in \intrange{0}{n}$. As such, Proposition~\ref{prop:asyvar} characterises the asymptotic local effect of the random-weights on the variance as $s_p$. 
\end{remark}
\begin{proposition}\label{prop:varbound} 
Consider the same models as Proposition~\ref{prop:asyvar} and let $\sigma_n^{2}(\varphi)$ be the asymptotic variance of $\gamma_n^{N}(\varphi)$ for measurable $\varphi$.
If, for some $0 < C < \infty$, $\sup_{p \in \intrange{0}{n}}\sup_{x_p \in \mathsf{X}}s_p(x_p) \leq C+1$, i.e. the relative variance is uniformly bounded, then
\begin{enumerate}[label=(\roman*)]
    \item $\bar{\sigma}_n^{2}(\varphi) \leq (C+1) \sigma_n^{2}(\varphi)$ if $\varphi$ is a centred function, i.e. $\eta_n(\varphi) = 0$,
    \item $\bar{\sigma}_n^{2}(1) \leq (C+1) \sigma_n^{2}(1) + nC$.
\end{enumerate}

\end{proposition}
Proofs for Proposition~\ref{prop:asyvar} and \ref{prop:varbound} are given in the supplementary materials.

\begin{remark}
Proposition~\ref{prop:varbound} demonstrates that the variance of expectations from particle approximations using the random-weight model are well-behaved if the relative variance of the random-weight potential function is also well-behaved. The upper bound for the variance of the normalising constant (ii) increases linearly in $C$ and $n$.
\end{remark}
\begin{remark}
For a $\psi$-twisted model we note that the twisted model is itself a Feynman--Kac model with random weights.  Then for some $C>0$, we can deduce the requirements
 $\var[\{M_1^{\tilde{N}}(\psi_{1})(x) M_{0}^{\tilde{N}}(\psi_0)\}/\{M_1(\psi_{1})(x) M_{0}(\psi_0)\}] \leq C +1$ and $\var\{M_p^{\tilde{N}}(\psi_p)(x)/M_p(\psi_p)(x)\} \leq C +1$ for $p \in \intrange{2}{n}$ and $x \in \mathsf{X}$.
\end{remark}

\section{Example applications}

\subsection{Linear Gaussian hidden Markov model}\label{sec:exlgm}

The performance of the proposed method was tested on a linear Gaussian hidden Markov model. The normalising constant for this model can be calculated exactly to aid comparisons between particle filters. We compare some model and algorithmic settings of the Monte Carlo-based twisted-model particle filter (henceforth TPF) and twisting learning procedure with reference to the bootstrap particle filter (BPF). We run two comparative BPFs, one with approximately equivalent computation cost (Markov kernel simulations) and the other with approximately equivalent memory cost (number of particles). These illustrate the extremes of the computation performance considerations.  We also run a TPF with optimal twisting functions tempered to have an equivalent acceptance rate to the TPF with learnt twisting functions. The results below are in terms of the mean square error of the log-normalising constant relative to this tempered optimal TPF (rel-MSE). Further details and results are given in the supplementary materials.

We considered a 3-dimensional hidden Markov model with time series length $n = 200$ and observation variance $\sigma^2_G \in \{0.25,1\}$. We used three iterations of Algorithm~\ref{alg:learntwist} and $\tilde{N} \in \{25,50,100\}$ samples to estimate the twisted potential function. The simulation demonstrated that our learning procedure performed just as well as the using tempered optimal TPF (rel-MSE 0.7--1.3), significantly better than the memory-equivalent BPF (26--101), and worse than the computation-equivalent BPF (0.02--0.1).

\subsection{A stochastic volatility model}\label{sec:exsvm}
We now illustrate the methodology on the root square stochastic volatility model \citep[see][for example]{martin2019auxiliary} for a $n=2000$ time series. Stochastic volatility models are popular in modelling financial return data \citep{heston1993closed}. Here we use partial analytical twisting as the Markov kernel can be exponentially tilted. Further details are given in the supplementary materials. 
Figure~\ref{fig:svm} shows that our method (with 2 learning iterations) outperforms both memory- and computation-equivalent BPFs (``=S'' and ``=C'' respectively) as the estimated normalising constant is much closer to the true value. 

\section{Discussion}
The results of Section~\ref{sec:exlgm} suggest that if storage limitations exist or large memory usage will reduce computational performance, then the rejection-based TPFs are useful in practice. This is likely to be the case for when the state-space requires high memory usage, such as in the case of a high dimensional latent space, long time series (i.e. large~$n$), or both. The use of a twisting function improves the particle filter but is countered by the additional cost of the rejection sampler and added variance from the random weights. Determining the conditions for improvement by using Monte Carlo twisting is an interesting question left for future research. 

Section~\ref{sec:exsvm} demonstrated that partial analytical twisting can substantially improve the normalising constant estimate by improving the acceptance rate of the rejection sampler without having to alter the learnt twisting functions. This comes at the expense of applicability, compared to the rejection-only method, but nonetheless allows twisting functions to be used in new applications.
\begin{figure}[ht]
    \centering
    \includegraphics[width=0.75\textwidth]{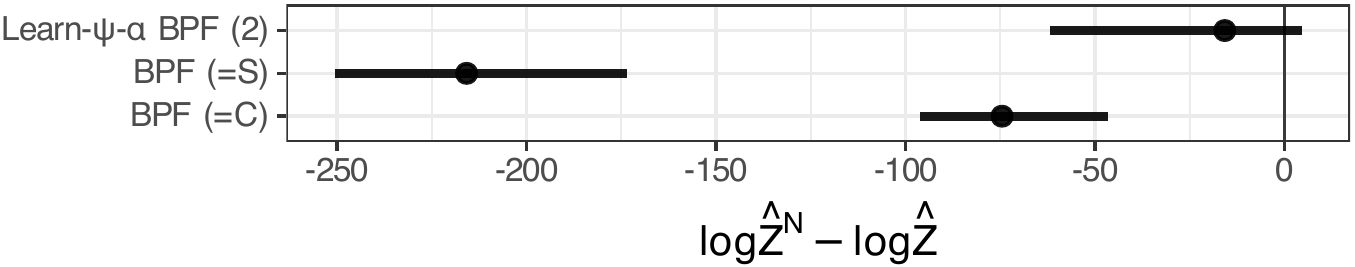}
    \caption{Normalising constant estimate error (log scale, median and 95\% interval) for the stochastic volatility model.}
    \label{fig:svm}
\end{figure}

\subsection*{Acknowledgement}
JJB, CD and AL were supported by the Australian Research Council (DP200102101). AL was also supported by the EPSRC (CoSInES, EP/R034710/1).

\subsection*{Supplementary material}\label{sec:supp}
Supplementary material are contained in the appendices for clarity. 

\bibliographystyle{unsrtnat}
\bibliography{references}

\begin{thebibliography}{17}
\providecommand{\natexlab}[1]{#1}
\providecommand{\url}[1]{\texttt{#1}}
\expandafter\ifx\csname urlstyle\endcsname\relax
  \providecommand{\doi}[1]{doi: #1}\else
  \providecommand{\doi}{doi: \begingroup \urlstyle{rm}\Url}\fi

\bibitem[Doucet and Johansen(2011)]{doucet2011tutorial}
Arnaud Doucet and Adam~M Johansen.
\newblock A tutorial on particle filtering and smoothing: Fifteen years later.
\newblock In Dan Crisan and Boris Rozovskii, editors, \emph{The {O}xford
  handbook of nonlinear filtering}, chapter 8.3. Oxford University Press,
  Oxford, 2011.

\bibitem[Andrieu and Roberts(2009)]{andrieu2009pseudo}
Christophe Andrieu and Gareth~O Roberts.
\newblock The pseudo-marginal approach for efficient {M}onte {C}arlo
  computations.
\newblock \emph{The Annals of Statistics}, 37\penalty0 (2):\penalty0 697--725,
  2009.

\bibitem[Andrieu et~al.(2010)Andrieu, Doucet, and
  Holenstein]{andrieu2010particle}
Christophe Andrieu, Arnaud Doucet, and Roman Holenstein.
\newblock Particle {Markov} chain {Monte} {Carlo} methods.
\newblock \emph{Journal of the Royal Statistical Society: Series B (Statistical
  Methodology)}, 72\penalty0 (3):\penalty0 269--342, 2010.

\bibitem[Lin et~al.(2013)Lin, Chen, and Liu]{lin2013lookahead}
Ming Lin, Rong Chen, and Jun~S Liu.
\newblock Lookahead strategies for sequential {Monte Carlo}.
\newblock \emph{Statistical Science}, 28\penalty0 (1):\penalty0 69--94, 2013.

\bibitem[Bresler(1986)]{bresler1986two}
Yoram Bresler.
\newblock Two-filter formulae for discrete-time non-linear {Bayesian}
  smoothing.
\newblock \emph{International Journal of Control}, 43\penalty0 (2):\penalty0
  629--641, 1986.

\bibitem[Briers et~al.(2010)Briers, Doucet, and Maskell]{briers2010smoothing}
Mark Briers, Arnaud Doucet, and Simon Maskell.
\newblock Smoothing algorithms for state--space models.
\newblock \emph{Annals of the Institute of Statistical Mathematics},
  62\penalty0 (1):\penalty0 61--89, 2010.

\bibitem[Guarniero et~al.(2017)Guarniero, Johansen, and
  Lee]{guarniero2017iterated}
Pieralberto Guarniero, Adam~M Johansen, and Anthony Lee.
\newblock The iterated auxiliary particle filter.
\newblock \emph{Journal of the American Statistical Association}, 112\penalty0
  (520):\penalty0 1636--1647, 2017.

\bibitem[Heng et~al.(2020)Heng, Bishop, Deligiannidis, and
  Doucet]{heng2020controlled}
Jeremy Heng, Adrian~N Bishop, George Deligiannidis, and Arnaud Doucet.
\newblock Controlled sequential {M}onte {C}arlo.
\newblock \emph{The Annals of Statistics}, 48\penalty0 (5):\penalty0
  2904--2929, 2020.

\bibitem[Gordon et~al.(1993)Gordon, Salmond, and Smith]{gordon1993novel}
Neil~J Gordon, David~J Salmond, and Adrian~FM Smith.
\newblock Novel approach to nonlinear/non-{G}aussian {B}ayesian state
  estimation.
\newblock \emph{IEE Proceedings F (Radar and Signal Processing)}, 140\penalty0
  (2):\penalty0 107--113, 1993.

\bibitem[Kong et~al.(1994)Kong, Liu, and Wong]{kong1994sequential}
Augustine Kong, Jun~S Liu, and Wing~Hung Wong.
\newblock Sequential imputations and {B}ayesian missing data problems.
\newblock \emph{Journal of the American Statistical Association}, 89\penalty0
  (425):\penalty0 278--288, 1994.

\bibitem[Liu and Chen(1995)]{liu1995blind}
Jun~S Liu and Rong Chen.
\newblock Blind deconvolution via sequential imputations.
\newblock \emph{Journal of the American Statistical Association}, 90\penalty0
  (430):\penalty0 567--576, 1995.

\bibitem[Del~Moral et~al.(2012)Del~Moral, Doucet, Jasra,
  et~al.]{del2012adaptive}
Pierre Del~Moral, Arnaud Doucet, Ajay Jasra, et~al.
\newblock On adaptive resampling strategies for sequential {M}onte {C}arlo
  methods.
\newblock \emph{Bernoulli}, 18\penalty0 (1):\penalty0 252--278, 2012.

\bibitem[Beskos et~al.(2006)Beskos, Papaspiliopoulos, Roberts, and
  Fearnhead]{beskos2006exact}
Alexandros Beskos, Omiros Papaspiliopoulos, Gareth~O Roberts, and Paul
  Fearnhead.
\newblock Exact and computationally efficient likelihood-based estimation for
  discretely observed diffusion processes (with discussion).
\newblock \emph{Journal of the Royal Statistical Society: Series B (Statistical
  Methodology)}, 68\penalty0 (3):\penalty0 333--382, 2006.

\bibitem[Fearnhead et~al.(2010)Fearnhead, Papaspiliopoulos, Roberts, and
  Stuart]{fearnhead2010random}
Paul Fearnhead, Omiros Papaspiliopoulos, Gareth~O Roberts, and Andrew Stuart.
\newblock Random-weight particle filtering of continuous time processes.
\newblock \emph{Journal of the Royal Statistical Society: Series B (Statistical
  Methodology)}, 72\penalty0 (4):\penalty0 497--512, 2010.

\bibitem[Martin et~al.(2019)Martin, McCabe, Frazier, Maneesoonthorn, and
  Robert]{martin2019auxiliary}
Gael~M Martin, Brendan~PM McCabe, David~T Frazier, Worapree Maneesoonthorn, and
  Christian~P Robert.
\newblock Auxiliary likelihood-based approximate {B}ayesian computation in
  state space models.
\newblock \emph{Journal of Computational and Graphical Statistics}, 28\penalty0
  (3):\penalty0 508--522, 2019.

\bibitem[Heston(1993)]{heston1993closed}
Steven~L Heston.
\newblock A closed-form solution for options with stochastic volatility with
  applications to bond and currency options.
\newblock \emph{The review of financial studies}, 6\penalty0 (2):\penalty0
  327--343, 1993.

\bibitem[Del~Moral(2004)]{del2004feynman}
Pierre Del~Moral.
\newblock \emph{Feynman-{K}ac Formulae}.
\newblock Springer-Verlag, New York, 2004.

\end{thebibliography}

\appendix

\section*{Appendices}

\section{Proof of Proposition~\ref{prop:acceptrate}}
\subsection*{Preliminaries}

We begin by proving some fundamental results for twisted Feynman--Kac models.

\begin{proposition}\label{prop:twistedmarginals}
Consider a $\psi$-twisted model $(M_{0:n}^{\psi}, G_{0:n}^{\psi})$. The predictive marginal measures can be written as 
\begin{align*}
\gamma_0^{\psi}(\dd x_0) &= \gamma_0(\dd x_0)\frac{ \psi_0(x_0)}{M_0(\psi_0)} \\
    \gamma_{p}^{\psi}(\dd x_p) &= \gamma_{p}(\dd x_p) \psi_p(x_p), \quad p \in \intrange{1}{n}.
\end{align*}
Whilst the marginal updated measures are
\begin{align*}
    \hat{\gamma}_{p}^{\psi}(\dd x_p) &= \hat\gamma_{p}(\dd x_p)M_{p+1}(\psi_{p+1})(x_{p}), \quad p \in \intrange{0}{n-1} \\
    \hat{\gamma}_{n}^{\psi}(\dd x_n) &= \hat\gamma_{n}(\dd x_n).
\end{align*}
\end{proposition}
\begin{proof}
For the case of $p=0$, note that $M_0^{\psi}(\dd x_0) = M_0(\dd x_0) \psi_0(x_0) / M_0(\psi_0)$, $\gamma_0^{\psi} = M_0^{\psi}$, and~${\gamma_0 = M_0}$ by definition which gives the result. For the cases $p > 0$, first note that
\begin{equation}\label{eq:potxmut}
\begin{split}
    G_{0}^{\psi}( x_{0})M_{1}^{\psi}(x_{0},\dd x_1) &= G_{0}( x_{0})M_{1}(x_{0},\dd x_1) \frac{\psi_1(x_1)}{\psi_0(x_{0})} M_0(\psi_0)\\
    G_{p-1}^{\psi}( x_{p-1})M_{p}^{\psi}(x_{p-1},\dd x_p) &= G_{p-1}( x_{p-1})M_{p}(x_{p-1},\dd x_p) \frac{\psi_p(x_p)}{\psi_p(x_{p-1})} 
\end{split}
\end{equation} for $p \in \intrange{2}{n}$ using the definitions of the twisted potential and mutation kernel then simplifying. Now, for the case of $p=1$
the definition of the marginal measures \eqref{eq:mmrecur} and twisted marginals \eqref{eq:potxmut} yields
\begin{align*}
    \gamma_{1}^{\psi}(S) &= \int_{\mathsf{X}} \gamma_{0}^{\psi}(\dd x_{0}) G_{0}^{\psi}(x_{0})M_{1}^{\psi}(x_{0}, S) \\
     &= \int_{\mathsf{X}} \gamma_{0}(\dd x_{0}) G_{0}(x_{0})\int_{S} M_{1}(x_{0}, \dd x_1)\psi_{1}(x_1) ~\text{for}~ S \in \mathcal{X} 
\end{align*}
implying that $\gamma_{1}^{\psi}(\dd x_1) = \gamma_{1}(\dd x_1)\psi_1(x_1)$ as required. Now assume $\gamma_{p}^{\psi}(\dd x_p) = \gamma_{p}(\dd x_p)\psi_p(x_p)$ is true for some $p \in \intrange{1}{n-1}$ and consider the case of $p+1$. We again combine \eqref{eq:mmrecur} and \eqref{eq:potxmut} to yield
\begin{align*}
    \gamma_{p+1}^{\psi}(S)
     &= \int_{\mathsf{X}} \gamma_{p}^{\psi}(\dd x_{p}) \frac{G_{p}(x_{p})}{\psi_{p}(x_{p})} \int_{S} M_{p+1}(x_{p}, \dd x_{p+1})\psi_{p+1}(x_{p+1}) \\
     &= \int_{\mathsf{X}} \gamma_{p}(\dd x_{p}) G_{p}(x_{p}) \int_{S} M_{p+1}(x_{p}, \dd x_{p+1})\psi_{p+1}(x_{p+1}) ~\text{for}~ S \in \mathcal{X}
\end{align*}
implying that $\gamma_{p+1}^{\psi}(\dd x_{p+1}) = \gamma_{p+1}(\dd x_{p+1})\psi_{p+1}(x_{p+1})$. Therefore by induction $\gamma_{p}^{\psi}(\dd x_{p}) = \gamma_{p}(\dd x_{p})\psi_p(x_{p})$ for $p \in \intrange{2}{n}$ also.  

For the updated measures $\hat{\gamma}_{p}^{\psi}(\dd x_p) = \gamma_{p}^{\psi}(\dd x_p) G_{p}^{\psi}(x_p)$ for $p \in \intrange{0}{n}$ by definition. Combining with the definition of the twisted potential functions in \eqref{eq:twistpot} and the previous results for the predictive marginal measures we have
\begin{align*}
    \hat{\gamma}_{0}^{\psi}(\dd x_0) &= \frac{\gamma_{0} (\dd x_0)\psi_0(x_0)}{M_0(\psi_0)}  \frac{G_{0}(x_0)}{\psi_{0}(x_{0})}M_1(\psi_{1})(x_{0})M_{0}(\psi_0) \\
    &= \hat\gamma_{0}(\dd x_0) M_1(\psi_{1})(x_{0}),\\
    \hat\gamma_{p}^{\psi}(\dd x_p) &= \gamma_{p}(\dd x_p) \psi_p(x_p) \frac{G_p(x_p)}{\psi_p(x_p)}M_{p+1}(\psi_{p+1})(x_{p}) \\
    &= \hat\gamma_{p}(\dd x_p)M_{p+1}(\psi_{p+1})(x_{p}) ~\text{for}~ p \in \intrange{1}{n-1},~\text{and}\\
    \hat{\gamma}_{n}^{\psi}(\dd x_n) &= \gamma_{n} (\dd x_n)\psi_n(x_n)  \frac{G_{n}(x_n)}{\psi_{n}(x_{n})} \\
     &= \hat\gamma_{n} (\dd x_n)
\end{align*}
as required.
\end{proof}

\begin{proposition}\label{prop:mutationtwist}
Consider a $\psi$-twisted model  $(M_{0:n}^{\psi}, G_{0:n}^{\psi})$ and functions~${\phi_p: \mathsf{X} \rightarrow (0, \infty)}$ for~${p \in \intrange{1}{n}}$ then
\begin{align*}
    M_0(\psi_0 \cdot \phi_0) &= M_0(\psi_0)M^{\psi}_0(\phi_{0})  \\
    M_{p}(\psi_p\cdot \phi_p)(x_{p-1}) &= M_{p}(\psi_{p})(x_{p-1}) M_{p}^{\psi}(\phi_{p})(x_{p-1}), \quad p \in \intrange{1}{n}.
\end{align*}
\end{proposition}
\begin{proof}
By definition of the twisted mutation kernels applied to $\phi_p$ are
\begin{align*}
    M^{\psi}_{0}(\phi_0) &= \int_{\mathsf{X}}\frac{M_{0}(\dd x_0)\psi_{0}(x_0)}{M_0(\psi_{0})}\phi_0(x_0) = \frac{M_0(\psi_{0} \cdot \phi_0)}{M_0(\psi_{0})} ~\text{and}\\
    M^{\psi}_{p}(\phi_p)(x_{p-1}) &= \int_\mathsf{X} \frac{M_{p}(x_{p-1}, \dd x_{p})\psi_p(x_p)}{M_p(\psi_{p})(x_{p-1})}\phi_p(x_p) = \frac{M_{p}(\psi_p\cdot\phi_p)(x_{p-1})}{M_p(\psi_{p})(x_{p-1})}, \quad p \in \intrange{1}{n}.
\end{align*}
Rearranging these equations gives the result.
\end{proof}
The following result was also stated in \citet{heng2020controlled}.
\begin{proposition}\label{prop:twisttwist}
Consider a $\psi$-twisted model $(M_{0:n}^{\psi}, G_{0:n}^{\psi})$. Twisting this model with functions~${\phi_p: \mathsf{X} \rightarrow (0, \infty)}$ for~${p \in \intrange{1}{n}}$ results in a $(\psi\cdot \phi)$-twisted Feynman--Kac model, that is $(M_{0:n}^{\psi\cdot \phi}, G_{0:n}^{\psi\cdot \phi})$.
\end{proposition}
\begin{proof}
By definition we have 
\begin{equation*}
     (M^{\psi}_{0})^{\phi}(\dd x_0) = \frac{M^{\psi}_{0}(\dd x_0)\phi_{0}(x_0)}{M^{\psi}_0(\phi_{0})} 
    = \frac{M_{0}(\dd x_0)\psi_{0}(x_0)\phi_{0}(x_0)}{M_0(\psi_{0}) M^{\psi}_0(\phi_{0})}
\end{equation*}
and using Proposition~\ref{prop:mutationtwist} we can see $(M^{\psi}_{0})^{\phi}(\dd x_0) = M^{\psi\cdot\phi}_{0}(\dd x_0)$. Likewise, by definition
\begin{equation*}
    (M^{\psi}_{p})^{\phi}(x_{p-1}, \dd x_{p}) = \frac{M^{\psi}_{p}(x_{p-1}, \dd x_{p})\phi_p(x_p)}{M^{\psi}_p(\phi_{p})(x_{p-1})} = \frac{M_p(x_{p-1}, \dd x_{p})\psi_p(x_p)\phi_p(x_p)}{M_p(\psi_{p})(x_{p-1})M^{\psi}_p(\phi_{p})(x_{p-1})},
\end{equation*}
then with Proposition~\ref{prop:mutationtwist} we see that $(M^{\psi}_{p})^{\phi}(x_{p-1}, \dd x_{p}) = M^{\psi \cdot \phi}_{p}(x_{p-1}, \dd x_{p})$ for $p \in \intrange{1}{n}$.
As for the potential functions, by definition
\begin{align*}
    (G^{\psi}_{0})^{\phi}(x_0) &= G_{0}^{\psi}(x_0)\frac{M_1^{\psi}(\phi_{1})(x_{0})}{\phi_{0}(x_{0})}M_{0}^{\psi}(\phi_0)\\ &= G_{0}(x_0)\frac{M_1(\psi_{1})(x_{0})M_1^{\psi}(\phi_{1})(x_{0})}{\psi_{0}(x_{0})\phi_{0}(x_{0})}M_{0}(\psi_0)M_{0}^{\psi}(\phi_0)
\end{align*}
then with Proposition~\ref{prop:mutationtwist} $(G^{\psi}_{0})^{\phi}(x_0) = G^{\psi\cdot\phi}_{0}(x_0)$. Similarly, by definition
\begin{align*}
    (G^{\psi}_{p})^{\phi}(x_p) = G^{\psi}_{p}(x_p)\frac{M^{\psi}_{p+1}(\phi_{p+1})(x_{p})}{\phi_{p}(x_{p})} = G_{p}(x_p)\frac{M_{p+1}(\psi_{p+1})(x_{p})M^{\psi}_{p+1}(\phi_{p+1})(x_{p})}{\psi_{p}(x_{p})\phi_{p}(x_{p})}
\end{align*}
for $p \in \intrange{1}{n}$ and again with Proposition~\ref{prop:mutationtwist} we have $(G^{\psi}_{p})^{\phi}(x_p) = G^{\psi\cdot \phi}_{p}(x_p)$. Thus, the result holds as the mutation kernels and potential functions are equivalent.
\end{proof}

\subsection*{Proof of Proposition~\ref{prop:acceptrate}}\label{aproof:acceptrate}
\begin{proof}
Assuming $\sup_{x \in \mathsf{X}}\omega_p(x)\varrho_p(x)^{-1} = 1$ ensures that the (conditional) acceptance rates of the rejection samplers are 
$M_{p}^{\varrho}(\omega_{p} / \varrho_p)$ for $p \in \intrange{0}{n}$. These can be simplified to 
\begin{align*}
    M_{0}^{\varrho}(\omega_{0} / \varrho_0) &=   M_{0}(\omega_{0}) M_{0}(\varrho_0)^{-1} \\
    M_{p}^{\varrho}(\omega_{p} / \varrho_p)(x) &= M_{p}(\omega_{p})(x) M_{p}(\varrho_p)(x)^{-1}, \quad p \in \intrange{1}{n}, \quad x \in \mathsf{X}
\end{align*}
using their definitions. This gives the result for $p=0$.

Now we turn our attention to the marginal measures, for use in the case of $p > 0$. Let the twisting functions be decomposed as $\omega_p = \psi_p \cdot \phi_p$ where~$\phi_p = \omega_p / \psi_p$ for~${p \in \intrange{1}{n}}$. As such, the updated marginal measures of the $\omega$-twisted model
are $\hat{\gamma}_{p-1}^{\omega} = \hat\gamma_{p-1}^{\psi\cdot\phi}$. Then from Proposition~\ref{prop:twisttwist} in combination with Proposition~\ref{prop:twistedmarginals} we can deduce
\begin{equation*}
    \hat{\gamma}_{p-1}^{\omega}(\dd x_{p-1}) =  (\hat\gamma_{p-1}^{\psi})^\phi(\dd x_{p-1}) = \hat{\gamma}_{p-1}^{\psi}(\dd x_{p-1}) M_{p}^{\psi}(\phi_{p})(x_{p-1}).
\end{equation*}
for $p \in \intrange{1}{n}$. Hence the normalised measures are
\begin{equation}\label{eq:accconv}
    \hat{\eta}_{p-1}^{\omega}(\dd x_{p-1}) = \frac{\hat{\gamma}_{p-1}^{\psi}(\dd x_{p-1}) M_{p}^{\psi}(\phi_{p})(x_{p-1})}{\hat{\gamma}_{p-1}^{\psi}\{( M_{p}^{\psi}(\phi_{p})\}} 
    = \frac{\hat{\eta}_{p-1}^{\psi}(\dd x_{p-1}) M_{p}^{\psi}(\phi_{p})(x_{p-1})}{\hat{\eta}_{p-1}^{\psi}\{ M_{p}^{\psi}(\phi_{p})\}}
\end{equation}
for $p \in \intrange{1}{n}$. From Proposition~\ref{prop:mutationtwist} we can write the function~$M_{p}^{\psi}(\phi_{p})$ as~${M_{p}^{\psi}(\phi_{p}) = M_{p}(\omega_{p}) / M_{p}(\psi_{p})}$
and substitute this term throughout \eqref{eq:accconv} to gain 
\begin{equation*}
    \hat{\eta}_{p-1}^{\omega}(\dd x_{p-1}) = \frac{\hat{\eta}_{p-1}^{\psi}(\dd x_{p-1}) M_{p}(\omega_{p})(x_{p-1}) M_{p}(\psi_{p})(x_{p-1})^{-1}}{\hat{\eta}_{p-1}^{\psi}\left\{ M_{p}(\omega_{p}) / M_{p}(\psi_{p}) \right\}}.
\end{equation*}

Finally, we recall $M_p^\varrho(\omega_p / \varrho_p) = M_p(\omega_p) / M_p(\varrho_p)$, and calculate~$\hat{\eta}_{p-1}^{\omega}M_p^\varrho(\omega_p / \varrho_p)$ for the result.
\end{proof}

\section{Algorithms for controlling the acceptance rate}

Algorithm~\ref{alg:learntemper} uses a particle approximation for the expressions in Proposition~\ref{prop:acceptrate} to determine the temperature~$\beta_p$ that targets a desired acceptance rate.
\begin{algorithm}
\caption{Targeting an acceptance rate at time $p \in \intrange{1}{n}$ by tempering}
\label{alg:learntemper}
Given particles $\hat{\zeta}_{p-1}^{1:N}  \sim \hat\eta_{p-1}^{\psi}$ from a $\psi$-twisted model, target acceptance rate $\alpha_{\text{min}}$, proposed updated twisting function $\psi_p \cdot \lambda_p$, and number of  Monte Carlo samples $\tilde{N}$,
\begin{enumerate}
    \item Draw $\zeta_{p}^{i,j} \sim M_p(\hat{\zeta}_{p-1}^{i}, \cdot)$ for $i \in \intrange{1}{N}$ and $j \in \intrange{1}{\tilde{N}}$
    \item Let $a_p^{i}(\beta, q) =  \left(\sum_{j=1}^{\tilde{N}}(\psi_p \cdot \lambda_p)(\zeta_p^{i,j})^{\beta} \right)^q  \left( \sum_{j=1}^{\tilde{N}}\psi_p(\zeta_p^{i,j}) \right)^{-1}$ \item[] Let
    $\alpha_p^N(\beta) = \frac{\sum_{i=1}^{N} a_p^i(\beta, 2)}{\sum_{i=1}^{N} a_p^i(\beta, 1)}$, the approximate acceptance rate for twisting function~${\omega_p = (\psi_p \cdot \lambda_p)^{\beta}}$
    \item Calculate $\beta_p$ by
    \begin{enumerate}[label=(\roman*)]
        \item If $\alpha_p^N(1) \geq \alpha_{\text{min}}$ let $\beta_p =1$, otherwise
        \item Set $\beta_p\in(0,1)$ such that $\alpha_p^N(\beta_p) \approx \alpha_{\text{min}}$ with a line search
    \end{enumerate}
\end{enumerate}
\textbf{Output:} $\beta_p$
\end{algorithm}

For time $p=0$ we can draw samples from the initial distribution and choose the temperature using a line search akin to Step 3 of Algorithm~\ref{alg:learntemper}. 

Using a temperature parameter, such as $\beta$, is not the only possible choice of transformation to control the acceptance rate.  Applying a lower bound to the twisting functions can increase the acceptance rate, or an upper bound with appropriate scaling. However, the temperature method is the easiest to implement in the context of iteratively learning the twisting functions.  Moreover, some testing with lower bounds showed poor performance as the bounded twisting functions were quite different from the unbounded twisting functions with respect to the base path measure.

Algorithm~\ref{alg:maxacc} uses a particle approximation for the expressions in Proposition~\ref{prop:acceptrate} to determine the best function to use for partial twisting.
\begin{algorithm}
\caption{Maximising the acceptance rate at time $p \in \intrange{1}{n}$ by partial analytical twisting}
\label{alg:maxacc}
Given particles $\hat{\zeta}_{p-1}^{1:N}  \sim \hat\eta_{p-1}^{\psi}$ from a $\psi$-twisted model, proposed updated twisting function $\psi_p \cdot \lambda_p$, $\mathsf{R}$ a tractable family for the partial twisting function $\varrho$, and chosen number of $\tilde{N}$ Monte Carlo samples,
\begin{enumerate}
    \item Draw $\zeta_{p}^{i,j} \sim M_p(\hat{\zeta}_{p-1}^{i}, \cdot)$ for $i \in \intrange{1}{N}$ and $j \in \intrange{1}{\tilde{N}}$
    \item Let $a_p^{i}(q) =  \left(\sum_{j=1}^{\tilde{N}}(\psi_p \cdot \lambda_p)(\zeta_p^{i,j}) \right)^q  \left( \sum_{j=1}^{\tilde{N}}\psi_p(\zeta_p^{i,j}) \right)^{-1}$ and
    \item[] $h_p^i(\varrho) =  \left( \sum_{j=1}^{\tilde{N}}\varrho(\zeta_p^{i,j}) \right)^{-1}$.
    \item[] Let
    $\alpha_p^N(\varrho) = \frac{\sum_{i=1}^{N} a_p^i( 2)h_p^i(\varrho)}{\sum_{i=1}^{N} a_p^i(1)}$, the approximate acceptance rate for twisting function~${\omega_p = \psi_p \cdot \lambda_p}$
    \item Find $\varrho_p = \max_{\varrho \in \mathsf{R}} \alpha_p^N(\varrho)$
\end{enumerate}
\textbf{Output:} $\varrho_p$
\end{algorithm}

\section{Random-weight sequential Monte Carlo}

Consider a Feynman--Kac model with initial distribution $M_0(\dd x_0)$, Markov kernels $M_p(x_{p-1}, \dd x_p)$ defined on $(\mathsf{X}, \mathcal{X})$ for $p \in \intrange{1}{n}$, and potential functions~$G_p(x_p)$ for $p \in \intrange{0}{n}$. We define a random-weight Feynman--Kac model, with respect to the aforementioned target model as follows.
\begin{definition}[Random-weight model]
Define the Feynman--Kac model on an extended space $(\tilde{\mathsf{X}}^{n+1}, \tilde{\mathcal{X}}^{n+1})$ where $\tilde{\mathsf{X}} = \mathsf{X} \times \mathsf{U}$ and $\tilde{\mathcal{X}} = \mathcal{X} \otimes \mathcal{U}$ with $\tilde{x}_p = (x_p, u_p)$ such that the initial distribution is \begin{equation*}
    \tilde{M}_0(\dd \tilde{x}_0) = M_0(\dd x_0)K_0(x_0, \dd u_0) 
\end{equation*} 
and the Markov kernels are
\begin{equation*}
    \tilde{M}_p(\tilde{x}_{p-1}, \dd \tilde{x}_p) = M_p(x_{p-1}, \dd x_p)K_p(x_{p}, \dd u_p)~\text{for}~ p \in \intrange{1}{n}
\end{equation*} where the Markov kernels $K_p(x, \cdot)$ are defined on the measurable space $(\mathsf{U}, \mathcal{U})$ for $x \in \mathsf{X}$. The potential functions $\tilde{G}_p(\tilde{x}_p)$ for $p \in \intrange{0}{n}$ are such that 
\begin{equation}\label{eq:potunb}
    \int_{\mathsf{U}}\tilde{G}_p(\tilde{x}_p)K_p(x_p,\dd u_p) = G_p(x_p).
\end{equation}
\end{definition}
We first prove some preliminary results used in the proof of Proposition~\ref{prop:asyvar} which follows.
\begin{proposition}\label{prop:marginalrw}
The time $p$ marginal measures of the random-weight model satisfy
\begin{equation}
    \tilde{\gamma}_p(\dd \tilde{x}_p) = \gamma_p(\dd x_p)K_p(x_p, \dd u_p), \quad p \in \intrange{0}{n}. \label{eq:marginalrw}
\end{equation}
\end{proposition}
\begin{proof}
Let $S = A \times B \in \tilde{\mathcal{X}}$.
The case of $p = 0$ holds since $\tilde{\gamma}_0(S) = \int_{A}M_0(\dd x_0)K_0(x_0,B)$ and~$\gamma_0(\dd x_0) = M_0(\dd x_0)$. Now assume \eqref{eq:marginalrw} is true for some $p \in \intrange{0}{n-1}$, then for~${p+1}$, using \eqref{eq:potunb}, we have
\begin{align*}
    \tilde{\gamma}_{p+1}(S) &= \int_{\mathsf{X} \times \mathsf{U}} \tilde{\gamma}_{p}(\dd \tilde{x}_{p})\tilde{G}_{p}(\tilde{x}_{p})\tilde{M}_{p+1}(\tilde{x}_{p},S) \\
    &= \int_{\mathsf{X} \times \mathsf{U}} \gamma_{p}(\dd x_{p})K_p(x_{p},\dd u_p)\tilde{G}_{p}(\tilde{x}_{p})\tilde{M}_{p+1}(\tilde{x}_{p},S) \\
    &= \int_{\mathsf{X}} \gamma_{p}(\dd x_{p})G_{p}(x_{p})\tilde{M}_{p+1}(\tilde{x}_{p},S)
\end{align*}
and therefore $\tilde{\gamma}_{p+1}(\dd \tilde{x}_{p+1}) = \gamma_{p+1}(\dd x_{p+1})K_{p+1}(x_{p+1},\dd u_{p+1})$. As such, the proposition is true by induction.
\end{proof}

Proposition~\ref{prop:marginalrw} demonstrates that the effect of the new kernels $K_p$ is localised to the time $p$ marginal measure (of distribution after normalisation) due to the unbiasedness condition in \eqref{eq:potunb}. Moreover, if we consider the marginal of $x_p$, we have 
\begin{equation}\label{eq:pmarginalequality}
    \tilde{\gamma}_p(\varphi \otimes 1) = \gamma_p(\varphi)
\end{equation}
for measurable $\varphi$. As such, expectations involving $x_{0:n}$ with respect to either measure are equivalent, and the normalising constant for the random-weight model is equal to the target model. Moreover, a particle filter for the random-weight model will also generate unbiased estimates of the normalising constant as it is a Feynman--Kac model targeting the correct marginal measure.

Let $Q_p$ be non-negative kernels defined as 
\begin{equation}\label{eq:Qp}
    Q_p(x_{p-1},\dd x_{p}) = G_{p-1}(x_{p-1})M_p(x_{p-1}, \dd x_p), \quad p \in \intrange{1}{n}
\end{equation} 
then define products of these non-negative kernels (also non-negative kernels) as
\begin{equation}\label{eq:Qpn}
\begin{split}
Q_{p,n} &= Q_{p+1} \cdots Q_n, \quad p \in \intrange{0}{n-1}\\
Q_{n,n} &= \text{Id}
\end{split}
\end{equation}
where Id is the identity kernel. We consider the asymptotic variance in the central limit theorem for  $\gamma_n^N(\varphi)/\gamma_n(1)$, defined as
\begin{align}
        &\sigma^2_n(\varphi) \equiv \sum_{p=0}^{n} v_{p,n}(\varphi), ~\text{where} \label{eq:asymvar}\\
& v_{p,n}(\varphi) \equiv \frac{\gamma_p(1)\gamma_p\left(Q_{p,n}(\varphi)^2\right)}{\gamma_n(1)^2} - \eta_{n}(\varphi)^2
\end{align}
for bounded $\varphi$. The central limit theorem \citep[Chapter 9,][]{del2004feynman} states
\begin{equation*}
    N^{1/2}\left\{\frac{\gamma_n^N(\varphi) - \gamma_n(\varphi)}{\gamma_n(1)}\right\} \rightarrow \mathcal{N}(0, \sigma^2_n(\varphi))
\end{equation*}
as $n \rightarrow \infty$ where convergence is in distribution.

We wish to consider the effect of using an unbiased estimate of the potential function on the asymptotic variance. We will continue to refer to objects related to the random-weight model with the tilde notation, and without as the target Feynman--Kac model.

\begin{proposition}\label{prop:tQpn}
The random-weight model has $\tilde{Q}_{p,n}$ such that
\begin{equation}\label{eq:potremainder}
\begin{split}
    \tilde{Q}_{p,n}(\tilde{x}_p, \dd \tilde{x}_n) &= \frac{\tilde{G}_{p}(\tilde{x}_p)}{G_p(x_p)} R_{p,n}(x_p, \dd \tilde{x}_n)\\
    R_{p,n}(x_p, \dd \tilde{x}_n) &= Q_{p,n}(x_p,\dd x_n)K_n(x_n,\dd u_n)
\end{split}
\end{equation}
for $p \in \intrange{0}{n-1}$.
\end{proposition}
\begin{proof}
By definition $\tilde{Q}_{p}(\tilde{x}_{p-1}, \dd \tilde{x}_{p}) = \tilde{G}_{p-1}(\tilde{x}_{p-1})\tilde{M}_p(\tilde{x}_{p-1}, \dd \tilde{x}_p)$ and 
\begin{align*}
    \tilde{Q}_{n-1,n}(\tilde{x}_{n-1}, \dd \tilde{x}_n) &= \tilde{G}_{n-1}(\tilde{x}_{n-1})\tilde{M}_p(\tilde{x}_{n-1}, \dd \tilde{x}_n) \\
    &= \frac{\tilde{G}_{n-1}(\tilde{x}_{n-1})}{G_{n-1}(x_{n-1})} G_{n-1}(x_{n-1})M_p(x_{n-1}, \dd x_n)K_n(x_{n}, \dd u_n) \\
    &= \frac{\tilde{G}_{n-1}(\tilde{x}_{n-1})}{G_{n-1}(x_{n-1})} Q_{n-1,n}(x_{n-1}, \dd x_n)K_n(x_{n}, \dd u_n) \\
    &= \frac{\tilde{G}_{n-1}(\tilde{x}_{n-1})}{G_{n-1}(x_{n-1})}R_{n-1,n}(x_{n-1}, \dd \tilde{x}_n)
\end{align*}
as such \eqref{eq:potremainder} is true for $p = n-1$. Now, assume \eqref{eq:potremainder} is true for some $p \in \intrange{1}{n-1}$, and consider the case of $\tilde{Q}_{p-1,n}$, for $S \in \mathcal{X} \times \mathcal{U}$ we have
\begin{align*}
    \tilde{Q}_{p-1,n}(\tilde{x}_{p-1}, S) &= \tilde{G}_{p-1}(\tilde{x}_{p-1})\int_{\mathsf{X} \times \mathsf{U}}\tilde{M}_p(\tilde{x}_{p-1}, \dd \tilde{x}_p) \tilde{Q}_{p,n}(\tilde{x}_{p}, S) \\
    &=  \tilde{G}_{p-1}(\tilde{x}_{p-1})\int_{\mathsf{X} \times \mathsf{U}} M_p(x_{p-1}, \dd x_p)K_p(x_{p}, \dd u_p) \frac{\tilde{G}_{p}(\tilde{x}_p)}{G_p(x_p)} R_{p,n}(x_p, S) \\
     &=  \tilde{G}_{p-1}(\tilde{x}_{p-1})\int_{\mathsf{X}} M_p(x_{p-1}, \dd x_p) R_{p,n}(x_p, S) \\
     &=  \frac{\tilde{G}_{p-1}(\tilde{x}_{p-1})}{G_{p-1}(x_{p-1})} R_{p-1,n}(x_p, S)
\end{align*}
as required. Therefore, \eqref{eq:potremainder} is true for $p \in \intrange{0}{n-1}$ by induction.
\end{proof}
Proposition~\ref{prop:tQpn} is the tool that allows us to relate the asymptotic variance of the random-weight model to the target model. In particular, we derive an explicit form for the variance in Proposition~\ref{prop:asyvar}, which considers the effect of the random-weights on the $v_{p,n}$ terms.

\subsection*{Proof of Proposition~\ref{prop:asyvar}}\label{aproof:asyvar}

\begin{proof}
Noting that $\tilde{\gamma}_p(1\otimes 1)=\gamma_p(1)$ and $\tilde{\eta}_p(\varphi\otimes 1)=\eta_p(\varphi)$ from \eqref{eq:pmarginalequality}, the $\tilde{v}_{p,n}$ are
\begin{equation*}
    \tilde{v}_{p,n}(\varphi \otimes 1) = \frac{\gamma_p(1)\tilde{\gamma}_p\left(\tilde{Q}_{p,n}(\varphi \otimes 1)^2\right)}{\gamma_n(1)^2} - \eta_{n}(\varphi)^2
\end{equation*}
for $p \in \intrange{0}{n}$. The relative second moment is 
\begin{equation*}
    s_p(x_p) = E\left( \frac{\bar{G}_{p}(x_p)^2}{G_p(x_p)^2}\right)  = E\left(\frac{\tilde{G}_{p}\{(x_p, U_p)\}^2}{G_p(x_p)^2}\right)  = \int_{\mathsf{U}} K_p(x_p,\dd u_p) \frac{\tilde{G}_{p}(\tilde{x}_p)^2}{G_p(x_p)^2}
\end{equation*}
for $x_p \in \mathsf{X}$, from which we can write
\begin{equation*}
    \tilde{\gamma}_p\left(\tilde{Q}_{p,n}(\varphi \otimes 1)^2\right) = \int_{\mathsf{X}\times \mathsf{U}} \gamma_p(\dd x_p) s_p(x_p) Q_{p,n}(\varphi)(x_p)^2 
    =\gamma_p(s_p \cdot Q_{p,n}(\varphi)^2)
\end{equation*}
for $p \in \intrange{0}{n-1}$ using Proposition~\ref{prop:tQpn}. For $p = n$ noting that $\tilde{Q}_{n,n}$ and $Q_{n,n}$ are identity kernels,
\begin{equation*}
    \tilde{\gamma}_n\left(\tilde{Q}_{n,n}(\varphi \otimes 1)^2\right) = \tilde{\gamma}_n\left\{(\varphi \otimes 1)^2\right\} = \tilde{\gamma}_n(\varphi^2 \otimes 1) =\gamma_n(\varphi^2) = \gamma_n(s_n \cdot Q_{n,n}(\varphi)^2)
\end{equation*}
with $s_n(x_n) = 1$ for $x_n \in \mathsf{X}$.
Finally, we can show
\begin{equation*}
    \bar{v}_{p,n}(\varphi) \equiv \tilde{v}_{p,n}(\varphi \otimes 1) = \frac{\gamma_p(1)\gamma_p(s_p \cdot Q_{p,n}(\varphi)^2)}{\gamma_n(1)^2} - \eta_{n}(\varphi)^2
\end{equation*}
by substitution for $p \in \intrange{0}{n}$ as required. Noting that~$\bar\sigma_n^{2}(\varphi) = \tilde{\sigma}_n^{2}(\varphi \otimes 1)$ concludes the proof.
\end{proof}

\subsection*{Proof of Proposition~\ref{prop:varbound}}\label{aproof:varbound}
If $\sup_{p \in \intrange{0}{n}}\sup_{x_p \in \mathsf{X}}s_p(x_p) \leq C+1$ for some $0<C<\infty$ then 
\begin{equation}\label{eq:boundvpn}
    \bar{v}_{p,n}(\varphi) \leq (C+1)\frac{\gamma_p(1)\gamma_p( Q_{p,n}(\varphi)^2)}{\gamma_n(1)^2} - \eta_{n}(\varphi)^2
\end{equation}
for $p \in \intrange{0}{n}$ by Proposition~\ref{prop:asyvar}. Then noting that 
\begin{equation*}
    v_{p,n}(\varphi) = \frac{\gamma_p(1)\gamma_p( Q_{p,n}(\varphi)^2)}{\gamma_n(1)^2} - \eta_{n}(\varphi)^2
\end{equation*}
for the non-approximate model, we can write \eqref{eq:boundvpn} as 
\begin{equation*}
    \bar{v}_{p,n}(\varphi) \leq (C+1)\left\{v_{p,n}(\varphi) + \eta_{n}(\varphi)^2\right\} - \eta_{n}(\varphi)^2 =  (C+1)v_{p,n}(\varphi) + C\eta_{n}(\varphi)^2.
\end{equation*}
If $\varphi$ is a centred function then $\eta_{n}(\varphi) = 0$ and we can deduce that  $\bar{\sigma}_n^{2}(\varphi) \leq (C+1) \sigma_n^{2}(\varphi)$ from \eqref{eq:asymvar}. Whilst if $\varphi = 1$ then $\eta_{n}(\varphi) = 1$ and $\bar{\sigma}_n^{2}(1) \leq (C+1) \sigma_n^{2}(1) + nC$ similarly.

\section{Details for examples}

\subsection*{Section \ref{sec:exlgm} example details}

The hidden Markov model considered has initial distribution $M_0(\dd x_0) = \dNorm(x_0; \mu_0, \sigma_0^2 I_d)\dd x_0$  and mutation kernels $M_p(x_{p-1}, \dd x_{p}) = \dNorm(x_p; A x_{p-1}, \sigma_M^2 I_d)\dd x_p$ for $p \in \intrange{1}{n}$ where $\dNorm(x; m, \Sigma)$ denotes the multivariate Gaussian density at $x$ with mean $m$ and covariance matrix $\Sigma$. The potential functions are $G_p(x_p) = \dNorm(y_p; x_p, \sigma_G^2 I_d)$ for $p \in \intrange{0}{n}$ where~$x_p,y_p, \mu_0 \in \R^d$, $\sigma_0^2,\sigma_M^2,\sigma_G^2 > 0$, $I_d$ is the $d$ dimensional identity matrix, and $A \in \R^{d\times d}$.

We consider a model with $d=3$, $\mu_0 = [1~1~ 1]^{\top}$, $\sigma_0^2 = \sigma_M^{2} = 1$, and the elements of matrix $A$ are such that $A_{i,j} = a^{\vert i-j\vert + 1}$ with $a = 0.42$. We choose $n = 200$ as the time series length. We use a dynamic resampling scheme based on effective sample size, resampling the particles only when the effective sample size drops below $N/2$ in a given iteration. 

The twisting functions are estimated with Algorithm~\ref{alg:learntwist} a total of $I = 3$ times where an isotropic exponential quadratic form is assumed for the class of twisting functions $\mathsf{H}_p$ for $p \in \intrange{0}{n-1}$. The minimum acceptance target is set to  $\alpha_\text{min} = 0.04$ for the first iteration to learn the twisting functions, followed by $0.02$ and $0.01$. As a safe guard against an excessively costly rejection sampler, we make a final transformation of the twisting functions as $\max(\psi_p(x_p), 5\times 10^{-4})$. This guarantees the acceptance rate is at least $5\times 10^{-4}$, though empirical results suggest Algorithm~\ref{alg:learntemper} is successful at targeting $\alpha_\text{min}$.

We test configurations with $N = 200$ particles for the twisted-model particle filters where ${\sigma_G^{2} \in \{0.25,1.0\}}$ and the number of Monte Carlo samples for the twisted potential approximation~${\tilde{N} \in \{25,50,100\}}$ with 100 repetitions for each setting. We use the same $\tilde{N}$ for Algorithms~\ref{alg:learntwist} and \ref{alg:learntemper}, though they need not be the same. The number of particles $N$ used for the computationally comparative BPF is calculated in each repetition by $\lceil N\times (3 \times R + 4 \times \tilde{N}) \rceil$ where $R$ is the mean number of iterations used by the rejection sampler over the 3 iterations. The additional $\tilde{N}$ is used to approximately account for the cost of learning the twisting functions, in particular the acceptance rate calculation. In contrast, the comparative storage BPF uses $N + \tilde{N}$ particles which results in a total storage cost of $n(N + \tilde{N})$ if the whole path is kept, whilst the twisted-model particle filters only require $nN + \tilde{N}$ storage space.

\begin{figure}
    \includegraphics[width=0.99\textwidth]{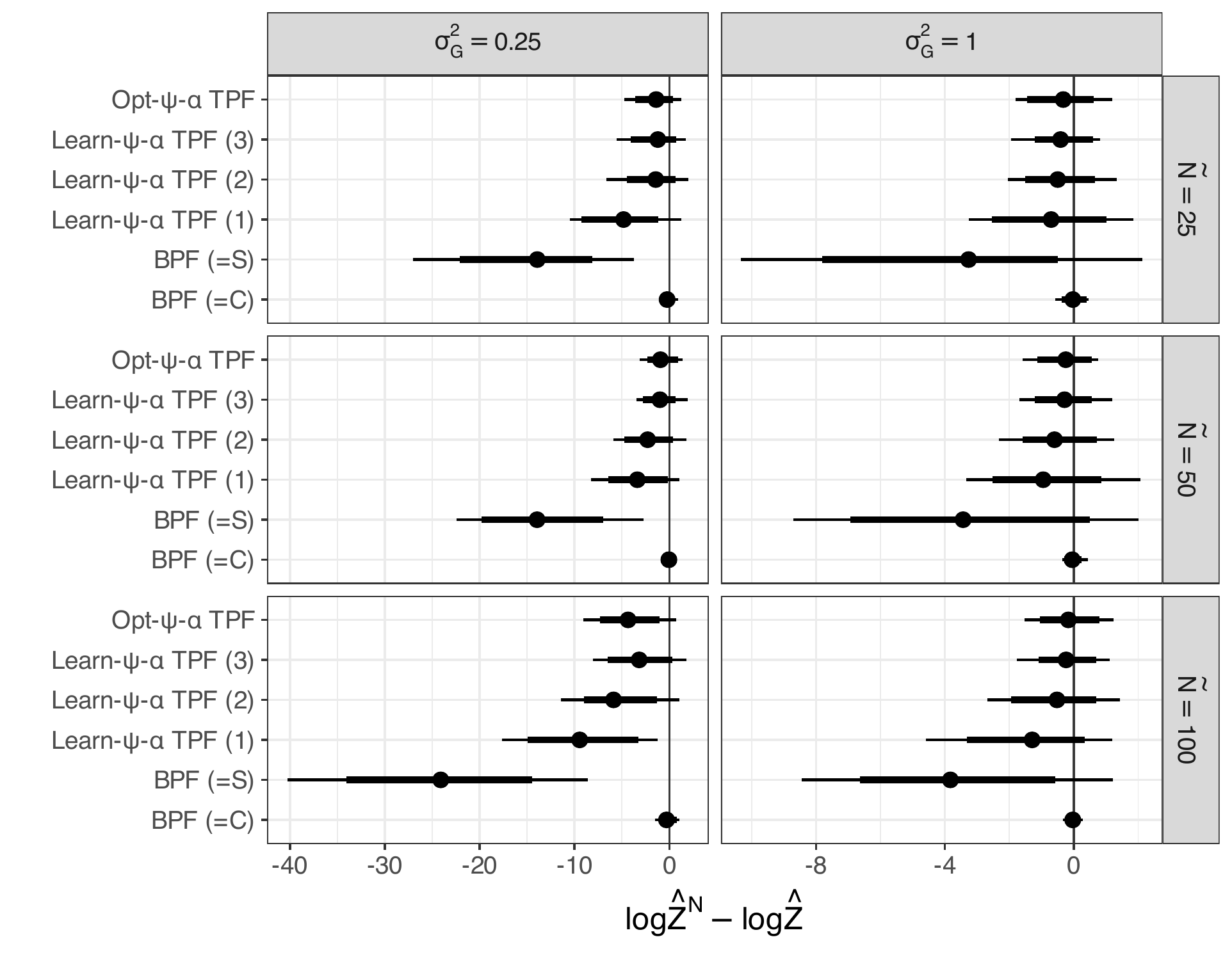}
    \caption{Error in normalising constant estimate (log scale, median, 90\% and 95\% intervals) for linear Gaussian hidden Markov model with varying observational variance $\sigma^2_{G}$ and Monte Carlo estimation of twisted potential functions $\tilde{N}$ over 100 repetitions. BPF (=C) and BPF (=S) are the comparative bootstrap particle filters in terms of computation and storage respectively. The particle filters using the optimal twisting functions (Opt-$\psi$-$\alpha$ TPF) and the learnt twisting functions after $I$ iterations (Learn-$\psi$-$\alpha$ TPF ($I$)) are tempered to target an acceptance rate $\alpha$. Values closer to 0 indicate better performance.}
    \label{fig:lghmm}
\end{figure}

\begin{table}
    \centering
    \begin{tabular}[t]{r r r r r}
\hline
$\sigma^2_G$ & $\tilde{N}$ & BPF (=C) & BPF (=S) & Learn-$\psi$-$\alpha$ TPF (3)\\
\hline
0.25 & 25 & 0.06 & 52.43 & 1.23\\

 & 50 & 0.05 & 100.82 & 1.20\\

 & 100 & 0.02 & 26.02 & 0.72\\
\hline

1.00 & 25 & 0.11 & 29.37 & 0.88\\

 & 50 & 0.10 & 43.14 & 1.31\\

 & 100 & 0.05 & 32.95 & 0.99\\
\hline
\end{tabular}
    \caption{Relative mean squared error of log-normalising constant, compared to tempered optimal twisted-model particle filter (Opt-$\psi$-$\alpha$). Tempering chosen to target acceptance rate $\alpha_\text{min} = 0.01$. Values closer to 0 indicate better performance.}
    \label{tab:relMSE-gauss}
\end{table}

\begin{table}
    \centering
    \begin{tabular}[t]{l r r r r r}
\hline
 &   &   &  \multicolumn{3}{c}{$\hat{Z}^N / \hat{Z}$} \\
 & $\tilde{N}$ & $\sigma^{2}_{G}$ & Mean & 10\% & 90\% \\
\hline
BPF (=C) & 25 & 0.25 & 0.94 & 0.40 & 1.54\\

BPF (=S) &   &   & 0.05 & 0.00 & 0.00\\



Learn-$\psi$-$\alpha$ TPF (3) &   &   & 1.58 & 0.02 & 2.02\\

Opt-$\psi$-$\alpha$ TPF &   &   & 0.63 & 0.03 & 1.43\\
\hline
BPF (=C) &  25 & 1.00 & 1.03 & 0.69 & 1.50\\

BPF (=S) &   &   & 2.68 & 0.00 & 0.61\\



Learn-$\psi$-$\alpha$ TPF (3) &   &   & 0.86 & 0.30 & 1.80\\

Opt-$\psi$-$\alpha$ TPF &   &   & 0.95 & 0.23 & 1.86\\
\hline
BPF (=C) & 50 & 0.25 & 0.98 & 0.59 & 1.38\\

BPF (=S) &   &   & 0.02 & 0.00 & 0.00\\



Learn-$\psi$-$\alpha$ TPF (3) &   &   & 0.95 & 0.06 & 1.95\\

Opt-$\psi$-$\alpha$ TPF &   &   & 0.84 & 0.10 & 2.47\\
\hline
BPF (=C) & 50 & 1.00 & 0.99 & 0.75 & 1.26\\

BPF (=S) &   &   & 0.74 & 0.00 & 1.65\\



Learn-$\psi$-$\alpha$ TPF (3) &   &   & 0.97 & 0.30 & 1.76\\

Opt-$\psi$-$\alpha$ TPF &   &   & 0.90 & 0.33 & 1.75\\
\hline
BPF (=C) & 100 & 0.25 & 1.04 & 0.33 & 2.14\\

BPF (=S) &   &   & 0.00 & 0.00 & 0.00\\



Learn-$\psi$-$\alpha$ TPF (3) &   &   & 0.66 & 0.00 & 1.40\\

Opt-$\psi$-$\alpha$ TPF &   &   & 0.36 & 0.00 & 0.35\\
\hline
BPF (=C) & 100 & 1.00 & 0.98 & 0.77 & 1.20\\

BPF (=S) &   &   & 0.41 & 0.00 & 0.56\\



Learn-$\psi$-$\alpha$ TPF (3) &   &   & 1.00 & 0.33 & 2.00\\

Opt-$\psi$-$\alpha$ TPF &   &  & 1.08 & 0.35 & 2.20\\
\hline
\end{tabular}
    \caption{Mean, 10\% and 90\% quantiles of normalising constant estimate ratio $\hat{Z}^N / \hat{Z}$ for linear Gaussian hidden Markov model with varying observational variance $\sigma^2_{G}$ and MC estimation of twisted potential functions $\tilde{N}$ over 100 repetitions. Values closer to 1 indicate better performance.}
    \label{tab:lghmm}
\end{table}

\subsection*{Section \ref{sec:exsvm} example details}
We consider the square root stochastic volatility model \citep[see][for example]{martin2019auxiliary} of the de-meaned logarithmic return $R_t$ defined by
\begin{equation}\label{eq:svm}
    R_t = X_t^{1/2}Z_t, \quad
    d X_t = (\phi_1 - \phi_2 X_{t}) d t +\phi_3 X_t^{1/2} dW_t
\end{equation}
where $Z_t \sim \dNorm(0,\sigma^2)$ and time-varying variance $X_t$ is defined by a stochastic differential equation (SDE) driven by Brownian motion $W_t$. Unit increments of the SDE have conditional distribution $(X_t \cd X_{t-1} = x) \sim \chi^{2}(2c, 2q+2, 2u)$ denoting a scaled non-central chi-squared distribution with degrees of freedom $2q+2$, non-centrality parameter $2u$, scaled by~$(2c)^{-1}$. The parameters are defined as $c = 2\phi_2 \phi_3^{-2}\{1 - \exp(-\phi_2)\}$, $u = c x \exp(-\phi_2)$, and $q = 2\phi_1 \phi_3^{-2} - 1$. The structural parameters are restricted to~$2\phi_1 \geq \phi_3^2$ and $0 < a \leq \phi_1, \phi_2, \phi_3 \leq b$ for some $a,b$.

The observations are the returns $R_t$ for $t \in \intrange{1}{n}$ used to define the potential functions. The density of the scaled non-central chi-square distribution, used for the mutation kernels, can be expressed with a Bessel function. These functions are tractable but can be costly to evaluate multiple times. Moreover, defining twisted versions of this stochastic volatility model introduces intractability as the mutation kernels can only be twisted with very simple functions. 

We use exponential quadratic twisting functions to define the twisted Feynman--Kac model. The non-central chi-square distribution has a compound Poisson-Gamma representation which we can use to perform part of the twisting analytically. We consider functions of the form
\begin{equation}
\begin{split}
    \phi(x) &= \max\left(\min \left(\phi^q(x), q_{\text{max}}\right), \alpha_{\text{min}}\right)\phi^l(x), \\
    \phi^q(x) &= \exp(a x^2 + b_q x + c), \quad
    \phi^l(x) = \exp(b_l x),
\end{split}
\label{eq:svtwist}
\end{equation}
for $x \in \R^+$ at each time point. The log-linear component of the twisting function, $\phi^l(x)$,  can be incorporated analytically into the non-central chi-square distribution as we show in the final section of Appendix~4. Thus we can improve the acceptance rates using partial analytical twisting without changing the current twisted Feynman--Kac model. A maximum value for the quadratic term, $q_\text{max}$, appears in \eqref{eq:svtwist} so that the quadratic coefficient can be either positively or negatively signed. The $q_\text{max}$ value is set to the maximum of the input data in the regression within Algorithm~\ref{alg:learntwist} and adjusted according to the quadratic-linear decomposition of \eqref{eq:svtwist} in Algorithm~\ref{alg:maxacc}. We use a minimum acceptance threshold of $\alpha_\text{min}  = 0.002$ for the example in this section.

We simulate a time series of length $n= 2000$ with $\phi_1 =0.1$, $\phi_2 = 0.5$, $\phi_3 = 0.1$, and $\sigma = 0.25$ according to the stochastic process defined in \eqref{eq:svm} with unit time increments. To test the competing particle filters we use parameters $\phi_1 =0.09$, $\phi_2 = 0.45$, $\phi_3 = 0.11$, and $\sigma = 0.25$. For the twisted-model particle filter (TPF) we use two learning iterations with $N_0 = 100$ particles and $N = 250$ for the final particle filter. The number of Monte Carlo samples for the twisted potential functions is $\tilde{N}_M = 2$, whilst the Monte Carlo samples for learning the twisting functions is $\tilde{N}_L = 5$.

All particle filters used a resampling threshold of $\kappa = 0.5$. We estimated the true normalising constant $\hat{Z}$ by using the mean of 10 repetitions of a BPF with $N = 20,000$ particles.

\subsection*{Exponentially tilted Poisson-Gamma distribution} \label{sec:etcsq}

We denote the (scaled) non-central chi-squared distribution by $\chi^{2}(k,\lambda,\beta)$ with degrees of freedom $k>0$, non-centrality parameter $\lambda > 0$ and rate $\beta > 0$. The probability density function can be written in several ways, here we use the compound Poisson-Gamma form, $(Y~\vert~V) \sim \text{Gam}(\alpha + V,\beta)$ and $V \sim \text{Pois}(\eta)$ with $\beta$ as the rate parameter. The (scaled) non-central chi-squared distribution is equivalent to the compound Poisson-Gamma distribution when $\alpha = k/2$ and $\eta = \lambda/2$.

The density of the compound Poisson-Gamma distribution is
\begin{equation*}
    f_X(x;\eta,\alpha, \beta) = \sum_{v=0}^{\infty}f_Y(x; \alpha + v, \beta)f_V(v;\eta)
\end{equation*}
where $f_Y$ is the Gamma density function and $f_V$ is the Poisson mass function.

We wish to tilt (or twist) $f_X$ by $\phi: \R^{+} \rightarrow \R^{+}$, an integrable function, resulting in the distribution with density
\begin{equation*}
    f_X^{\phi}(x;\eta,\alpha,\beta) = \frac{ f_X(x;\eta,\alpha,\beta)\phi(x) }{\mathcal{Z}_{X,\phi}}
\end{equation*}
where $\mathcal{Z}_{X,\phi} = \int_{\R^{+}} \phi(x) f_X(x;\eta,\alpha,\beta)  \dd x$.

Using the compound Poisson-Gamma form this can be written as
\begin{align*}
    f_X^{\phi}(x;\eta,\alpha, \beta) &= \frac{\phi(x)}{\mathcal{Z}_{X,\phi}}\sum_{v=0}^{\infty}f_Y(x; \alpha + v, \beta)f_V(v;\eta) \\
    &= \sum_{v=0}^{\infty}\frac{f_Y(x; \alpha + v, \beta)\phi(x)}{\mathcal{Z}_{Y,\phi}(v)}\frac{f_V(v;\eta)\mathcal{Z}_{Y,\phi}(v)}{\mathcal{Z}_{X,\phi}}.
\end{align*}

If $\phi(x) = \exp\{-b x\}$ then $\mathcal{Z}_{Y,\phi}(v) = \left(1 + b/\beta \right)^{-\alpha - v}$ for $b > -\beta$ per the moment-generating function of the Gamma distribution. The exponentially tilted Gamma distribution is itself a Gamma distribution with density $f_Y(x; \alpha + v, \beta +b )$. As for the Poisson density,
\begin{align*}
    \frac{f_V(v;\eta)\mathcal{Z}_{Y,\phi}(v)}{\mathcal{Z}_{X,\phi}} &= \mathcal{Z}_{X,\phi}^{-1} \frac{\eta^v \exp\{-\eta\}}{v!}\left(1 + b/\beta \right)^{-\alpha - v} \\
    &= \mathcal{Z}_{X,\phi}^{-1} \frac{\left(\frac{\eta\beta}{\beta + b}\right)^v \exp\left\{-\frac{\eta\beta}{\beta + b}\right\} }{v!} \exp\left\{\frac{\eta\beta}{\beta + b}-\eta\right\}\left(1 + b/\beta \right)^{-\alpha}.
\end{align*}
Therefore we have the exponentially tilted Poisson-Gamma given by
\begin{align*}
    (Y~\vert~V) &\sim \text{Gam}(\alpha + V,\beta + b) \\
    V &\sim \text{Pois}\left(\frac{\eta\beta}{\beta + b}\right)
\end{align*}
with normalising constant
$\mathcal{Z}_{X,\phi} = \exp\left\{\frac{\eta\beta}{\beta + b}-\eta\right\}\left(1 + b/\beta \right)^{-\alpha}$.

The expectation $E(Y)$ is
\begin{equation*}
    E(Y) = E(E(Y~\vert~V)) = \frac{\alpha + E(V)}{\beta + b} = \frac{\alpha}{\beta +b} + \frac{\eta\beta}{(\beta + b)^2}.
\end{equation*}

\end{document}